\documentclass[10pt,conference]{IEEEtran}
\usepackage{subfig}
\usepackage{bm}
\usepackage{mathrsfs}
\usepackage{textcomp}
\usepackage{epsfig}
\usepackage{graphicx}
\usepackage{epstopdf}
\usepackage{indentfirst}
\usepackage{amsmath}
\usepackage{amssymb}
\usepackage{array}
\usepackage{multirow}
\usepackage{graphicx}
\usepackage{color}

\usepackage{amsthm}
\newtheorem{theorem}{Theorem}

\newtheorem{proposition}[theorem]{Proposition}
\usepackage{times}
\usepackage[ruled,vlined]{algorithm2e}
\theoremstyle{definition}

\theoremstyle{remark}

\usepackage{cite}

\def\BibTeX{{\rm B\kern-.05em{\sc i\kern-.025em b}\kern-.08em T\kern-.1667em\lower.7ex\hbox{E}\kern-.125emX}}
\usepackage{graphicx}

\begin{document}
\title{A Deep Reinforcement Learning based Approach for NOMA-based Random Access Network with Truncated Channel Inversion Power Control}


\author{
   \IEEEauthorblockN{Ziru Chen\IEEEauthorrefmark{1}, Ran Zhang\IEEEauthorrefmark{2}, Lin X. Cai\IEEEauthorrefmark{1}, Yu Cheng\IEEEauthorrefmark{1}, and Yong Liu\IEEEauthorrefmark{3}}
    \IEEEauthorblockA{\IEEEauthorrefmark{1}Department of Electrical and Computer Engineering, Illinois Institute of Technology, Chicago, USA }
    \IEEEauthorblockA{\IEEEauthorrefmark{2} College of Engineering and Computing, Miami University, Oxford, OH, USA}
    \IEEEauthorblockA{\IEEEauthorrefmark{3} Department of Physics and Telecommunication Engineering, South China Normal University, Guangzhou, China}   
\IEEEauthorblockA{zchen71@hawk.iit.edu, zhangr43@miamioh.edu,  \{lincai, cheng\}@iit.edu,  yliu@m.scnu.edu.cn}
}

\maketitle

\begin{abstract}
As a main use case of 5G and Beyond wireless network, the ever-increasing machine type communications (MTC) devices pose critical challenges over MTC network in recent years. It is imperative to support massive MTC devices with limited resources. To this end, Non-orthogonal multiple access (NOMA) based random access network has been deemed as a prospective candidate for MTC network. In this paper, we propose a deep reinforcement learning (RL) based approach for NOMA-based random access network with truncated channel inversion power control. Specifically,  each MTC device randomly selects a pre-defined power level with a certain probability for data transmission. Devices are using channel inversion power control yet subject to the upper bound of the transmission power. Due to the stochastic feature of the channel fading and the limited transmission power, devices with different achievable power levels have been categorized as different types of devices.     
In order to achieve high throughput with considering the fairness between all devices, two objective functions are formulated. One is to maximize the minimum long-term expected throughput of all MTC devices, the other is to maximize the geometric mean of the long-term expected throughput for all MTC devices.
A Policy based deep reinforcement learning approach is further applied to tune the transmission probabilities of each device to solve the formulated optimization problems. Extensive simulations are conducted to show the merits of our proposed approach. 

\end{abstract}


\vspace{-11pt}
\section{Introduction}
\vspace{-3pt}
The numbers of machine-type communication (MTC) devices and the corresponding mobile data volume have grown rapidly with the development of smart metering, smart traffic surveillance, environmental monitoring, smart grid and other Internet of Things (IoT) applications. More than $14.7$ billion machine type devices are anticipated to be connected to the internet by the year of 2023~\cite{forecast2019cisco}. 
Due to the fact that the channel conditions are too costly to measure and update for low power machines, random access (RA) has attracted great attention for MTC networks~\cite{clazzer20195g}.  


Non-orthogonal multiple access (NOMA), as a promising technology in 5G and Beyond, allows more than one device sharing the same time-frequency resource block, which improves the spectral efficiency considerably~\cite{ni2019analysis}. Recently, the notion of NOMA is applied to slotted ALOHA system in order to achieve higher throughput for MTC network~\cite{chen2020optimizing,chen2021performance}. Specifically, the transmission power of devices need to be tuned according to the channel state information (CSI) to guarantee that the received signal strength equals to one of some predefined values at the receiver. By empowering each device to decide its transmission power with different probabilities, NOMA significantly improves the network throughput performance by resolving collisions via successive interference cancellation (SIC) technique. However, due to the power limitation, IoT devices may have various constraints in a realistic IoT network. For instance, some IoT devices may not be able to transmit in all power levels. In order to analyze this scenario, the author in~\cite{chen2021performance1} developed an analytical model, in which two type of devices has been considered and two power levels are available at the receiver side. Based on the analytical model, two algorithms have been proposed to find the optimal transmission probabilities to attain the maximum throughput and max-min fairness respectively. However, the analytical model requires the knowledge of the number of devices for each type, and the proposed algorithms cannot be applied to the case with more than two power levels are available in the system due to the complexity of formulation.


Machine learning (ML) plays an very important role in human life these days. Reinforcement learning (RL), as a promising solution to handle many ML problems, has been applied extensively in NOMA based MTC networks in recent years. In order to minimize random access channel collision, a Q-learning algorithm has been implemented for each MTC device to dynamically select RA slots and transmit power for its transmission~\cite{da2020noma}. The authors in~\cite{tran2021bler} extended the work in~\cite{da2020noma} by further considering the short-packet communication and imperfect successive interference cancellation. In~\cite{liu2021distributed}, the case with unsaturated traffic has been further considered.  
However, the proposed Q-learning methods in~\cite{da2020noma,tran2021bler,liu2021distributed} cannot handle the situation when the number of supported devices changes dynamically due to the limitation of the Q table. Besides, the fairness between devices cannot be guaranteed.
In order to tackle those problems, NOMA based slotted ALOHA scheme may better suit the MTC network.
The author in~\cite{khairy2021data} applied reinforcement learning method to an adaptive NOMA based p-persistent slotted ALOHA protocol. However, the reversed power control has not been considered which may cause the performance degradation.

In this work, we study the network throughput performance of the NOMA-based slotted ALOHA in MTC network.
To capture the realistic power constraints of IoT devices and stochastic wireless fading channels, truncated channel inversion power control is considered. We first analyze the power level design strategy in MTC network. With given power level design, devices in the network are categorized into different types, with some types of devices being capable to utilize all power levels, while other types of devices only using lower power levels.   
In order to guarantee the fairness between different types of devices, two optimization problems have been formulated, namely, to maximize the minimum long-term expected throughput and to maximize the geometric mean of the long-term expected throughput for devices in the network.
To solve the optimization problems effectively, a stateless deep RL learning approach has been proposed. To the best of our knowledge, our approach is the first to optimize network performance of the NOMA based slotted ALOHA in MTC network with truncated channel inversion power control, in which more than two types of devices exist in the network. Extensive simulations have been conducted to validate the performance of our proposed approach.

The remainder of the paper is organized as follows. Section~\ref{sec:system model} describes the system model. The power level design analysis is presented in Section~\ref{sec:powerlevelanalysis}. Optimization problems are formulated and a deep RL approach has been proposed to solve the optimization problem in Section~\ref{sec:algorithms}. The performance of our proposed approach has been validated and analyzed in Section~\ref{sec:numerical}, followed by concluding remarks and future work in Section~\ref{sec:conclusion}.

\label{sec:label}

\vspace{-9pt}
\section{System Model}  \label{sec:sys}
\begin{figure}[t]
\centering
\includegraphics[width=0.49\textwidth]{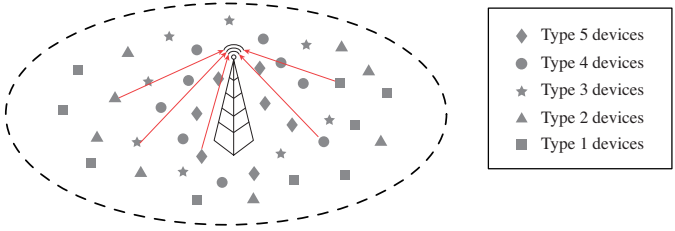}
\caption{System model of a NOMA based uplink MTC network.}
\label{illustration}
\end{figure}
\vspace{-4pt}
As depicted in Fig.~\ref{illustration}, a single cell uplink MTC network system has been considered, in which an access point (AP) is located at the center of a circular coverage area with a radio of $R$ meters, and multiple IoT devices are randomly distributed in its coverage area. The number of IoT devices may change dynamically to capture the mobility feature of the MTC network. 

During the uplink transmission, IoT devices send data to the AP by using NOMA based $p$-persistent slotted ALOHA protocol. 
Specifically, time is slotted. Power-domain NOMA is exploited to allow the AP to receive data from multiple IoT devices in each time slot. During each transmission, IoT devices can adjust its transmission power to ensure the received signal strength at the AP side belongs to a fixed predefined set $\mathcal{V}$ where $\mathcal{V}=\{V_1,V_2,\dots,V_M\}$, and $V_1<V_2<\dots<V_M$.   
Due to the fact that IoT devices have limited transmission power, truncated channel inversion power control has been considered in this scenario. This is, IoT devices in the coverage area may not be able to adjust their transmission power to achieve all power levels. Therefore, $M$ type of devices could exist in the system. i.e., The first type of devices can only utilize the power level $V_1$ as their received signal strength at the AP side. The second type of devices can use all power levels less than or equal to $V_2$ and etc. Thus, the $M$-th type of devices can exploit all $M$ power levels. Let $\mathcal{N} = \{\mathcal{N}_1,\mathcal{N}_2, \dots, \mathcal{N}_M \}$ be the set of IoT devices in the system, in which $\mathcal{N}_n$ denotes the set of $n$-th type of IoT devices and $1\leq n\leq M$. 

Similar to $p$-persistent slotted ALOHA protocol, a transmission probabilities matrix $\mathbb{P}$ has been introduced to guide the uplink transmission of IoT devices. The transmission probabilities matrix is updated by AP and will broadcast to all devices after every $T$ time slots. If we use $m$ to indicate the index of the power level, the probability of a $n$-th type of device to transmit by using $m$-th power level is denoted as $\tau_{n,m}$. Since the transmission probability for each type of devices is less than or equal to 1, we have $\sum_{m=1}^M\tau_{n,m}\leq 1, \forall n\leq M$.  And $\tau_{n,m} = 0, \forall n<m$ due to the fact that $n$-th type of devices cannot transmit with $m$-th power level when $n<m$. Thus, the transmission probabilities matrix $\mathbb{P}$ is a triangular matrix which can be written as
$$
\mathbb{P}=
\begin{bmatrix}
\tau_{1,1}& 0 &\cdots&0\\
\tau_{2,1}&\tau_{2,2}&\cdots&0\\
\vdots&\vdots&\ddots&\vdots\\
\tau_{M,1}&\tau_{M,2}&\cdots&\tau_{M,M}\\
\end{bmatrix}.
$$


Multiple devices may transmit simultaneously during each time slot in a random access network. If we use $\mathcal{N}^\prime_m$ to denote the set of devices transmitting by using $V_m$ as received power, the received signal at the AP side can be written as
\vspace{-3pt}
\begin{align}
    y=\sum_{V_m\in\mathcal{V}}\sum_{l\in\mathcal{N}^\prime_m}\sqrt{V_m}x_{l}+z,
\end{align}
\vspace{-3pt}
where $x_{l}$ denotes the transmitted signals of the device $l$, and $z$ is the background noise.

The AP can decode signals sequentially by applying SIC technique based on the descending order of the signal strength. Specifically, the AP starts the decoding from the signal with the highest receiving power under the interference from all other signals which is transmitting concurrently. Without loss of generality, we assume that the AP can decode the signal successfully only when the SINR of the decoding signal larger than or equal to a threshold $\Gamma$. Once the signal has been successfully decoded, it will be canceled by the AP. Thus, the rest signals would not be interfered by it. Throughput of the signal that has been decoded successfully is given by
\begin{align}
    \mathbf{Th} = \log_2(1+\mathbf{SINR}). \nonumber
\end{align}
In contrast, if the signal has not been decoded successfully, it cannot be canceled and will interfere the decoding of the following signals. The corresponding throughput of the signal is thus $0$. It is worth noting that since the AP decode signals on the descending order of the signal strength, once a signal with high power level cannot be decoded successfully, all following signals cannot be decoded successfully as well. It is possible that there are more than one devices transmitting by using the same power level. In this case, the AP decodes their signals sequentially in random order.

If all signals before the $i$-th decoding signal with the power level $V_m$ have been decoded and canceled successfully, the interference comes from all following signals that haven't been decoded.
In this case, the SINR of the signal can be written as
\begin{align}
    \mathbf{SINR}_m^{i}=\frac{V_m}{V_m(|\mathcal{N}^\prime_m|-i)+\sum_{x<m} V_x|\mathcal{N}^\prime_x|+\delta^2},
\end{align}
in which $|\cdot|$ denotes the number of devices in the corresponding device set and $i\leq|\mathcal{N}^\prime_m|$, $\delta^2$ is the normalized background noise.

\label{sec:system model}

\vspace{-9pt}
\section{ Power Level Analysis}  \label{sec:powerlevelanalysis}
In this section, we analyze the constraints of power levels after inversion power control. To take the advantage of NOMA, the predefined power level set $\mathcal{V}$ needs to satisfy that
\begin{align}\label{power_level_general}
    \frac{V_m}{\sum_{x<m}V_x+\delta^2}\geq \Gamma, \forall 1<m\leq M
\end{align}
and $V_1/\delta^2\geq \Gamma$. 
\vspace{-3pt}

\begin{proposition}
\label{lemma:power_level}
If we use $V_{max}$ to indicate the maximum achievable power level of the system, the maximum number of power levels we can have is given by
\begin{align}
    M = \lfloor \frac{\log(\frac{V_{max}}{\delta^2\Gamma})}{\log(1+\Gamma)}+1  \rfloor,
\end{align}
in which $\lfloor x \rfloor$ denotes the highest integer smaller than or equal to $x$.
\end{proposition}
\begin{proof}
To find the maximum number of power levels, the gap between all power levels should be as small as possible. 
Under this circumstance, equation~\eqref{power_level_general} can be rewritten as follow,
\begin{align}
    V_m&= \Gamma (\sum_{x<m}V_x+\delta^2) \nonumber \\
    &= \Gamma V_{m-1} + \Gamma (\sum_{x<m-1}V_x+\delta^2) \nonumber \\
    &= \Gamma V_{m-1} + V_{m-1}.
\end{align}
Thus, we have $V_{m}=V_{max}/(1+\Gamma)^{(M-m)}$.
To ensure $V_1\geq \delta^2\Gamma$,
\begin{align}
    \frac{V_{max}}{(1+\Gamma)^{(M-1)}}&\geq \delta^2\Gamma \nonumber \\
    \frac{V_{max}}{\delta^2\Gamma}&\geq (1+\Gamma)^{(M-1)} \nonumber \\
    \frac{\log(\frac{V_{max}}{\delta^2\Gamma})}{\log(1+\Gamma)}+1&\geq M.
\end{align}
\end{proof}
It is also possible to design power levels in other way as long as equation~\eqref{power_level_general} is satisfied, and the deep RL method we proposed can also solve the problem. However, in the following paper, we will focus on the case that $V_{m}=V_{max}/(1+\Gamma)^{(M-m)}$ which helps us to utilize more power levels for the MTC network.

\vspace{-9pt}
\section{Proposed RL Method} \label{sec:algorithms}
\vspace{-3pt}
\subsection{Optimization Problems}
Our goal is to maximize the network performance by tuning the transmission probabilities matrix $\mathbb{P}$. The most commonly used performance matrix for IoT network is the total expected throughput. However, total expected throughput and fairness are generally conflicting performance metrics in heterogeneous IoT networks. Specifically, devices belong to the same type can achieve long term fairness in a random access network, yet devices belong to different types use different transmission probabilities and achieve different expected throughput. The NOMA transmission may favor the type of devices which could help to achieve the highest total expected throughput and stop the transmission of other type of devices to avoid channel congestion. With considering the fairness, two optimization problems has been considered. The first objective function we considered in this paper is to achieve the max-min fairness of devices in the network. In this case, we formulate the decision problem of tuning $\mathbb{P}$ as a optimization problem
\begin{align}
(\mathbf{P1}) \qquad & \underset{\mathbb{P}}{\text{maximize}}
& \min \left(\mathbf{\overline{Th}}_{l}\right),  \qquad \forall l\in\mathcal{N},
\end{align}
in which $\mathbf{\overline{Th}}_{l}$ is the expected throughput of device $l$ over $T$ time slots.

Geometric mean of the expected throughput for all devices, on the other hand, is a performance matrix that also considered the fairness\cite{khairy2021data}. It is zero if any device in the network do not have chance to transmit. With the increasing of the geometric mean, we ensure that no device is starved without any chance to transmit, and the expected throughput of most devices are increasing. The corresponding optimization problem can be written as
\begin{align}
(\mathbf{P2}) \qquad & \underset{\mathbb{P}}{\text{maximize}}
& \left(\prod_{l\in \mathcal{N}} \mathbf{\overline{Th}}_{l}\right)^{1/|\mathcal{N}|},
\end{align}
where $|\mathcal{N}|$ is the number of devices in the IoT network.

The objective functions in $\mathbf{P1}$ and $\mathbf{P2}$ are mathematically intractable. Therefore, we propose a data-driven approach where a policy-based deep RL agent is applied at the AP to learn the transmission probabilities matrix $\mathbb{P}$ automatically.

\subsection{Deep RL basic}
The goal of a RL approach is to find an optimal strategy, i.e., a sequence of actions that maximizes the long-term expected accumulated discounted reward. Policy based RL methods, specifically, are well-known in addressing tasks with continuous action space. There are several policy based RL algorithms which has been developed recently, i.e., REINFORCE, trust region policy optimization (TRPO), deep deterministic policy gradient (DDPG), and proximal policy optimization (PPO). Among these algorithms, PPO draws great attention due to the fact that it is efficient, easy to be implemented and tuned~\cite{han2020deep}.

In policy based RL algorithms, the most commonly used estimator is given by 
\begin{align}
    L(\pi_{\theta_a}) = E_{\tau}[\pi_t(a_t|s_t;\theta_a)A_t(s_t,a_t)],
\end{align}
where $s_t$ and $a_t$ represent state and action at time $t$ respectively, $\pi_t(a_t|s_t;\theta_a)$ represents the policy at time $t$ and $\theta_a$ is the parameter of actor neural network which is used to generate policy $\pi_t(a_t|s_t;\theta_a)$. With given parameter $\theta_a$, the action can be generated by using Gaussian distribution,
\begin{align}
    \pi_t(a_t|s_t;\theta_a) = \frac{1}{\sqrt{2\pi\sigma(s_t,\theta_\sigma)^2}}\exp-\frac{(a_t-\mu(s_t,\theta_u)^2)}{2\sigma(s,\theta_\sigma)^2},
\end{align}
in which $\mu(s_t,\theta_u)$ and $\sigma(s,\theta_\sigma)$ are generated parameters from actor neural network. 
$A_t(s_t,a_t) = G_t-B_t$ is the advantage value where $G_t$ is the discounted future reward after time $t$ and $B_t$ is baseline.

It is worth noting that in PPO, data generated in previous episode can also be used to update current policy. In order to reuse the historical data, a clipping function is used to avoid large changes between current updated policy and the old policy. The clip function is given by,
\begin{align}
clip(y, 1-\epsilon, 1+\epsilon)\!\!=\!\!
\begin{cases}
y,& \mbox{if $1-\epsilon < y < 1+\epsilon$},\\
1-\epsilon,& \mbox{if $y\leq 1-\epsilon$},\\
1+\epsilon,& \mbox{if $y\geq 1+\epsilon$}. 
\end{cases}
\end{align}
The changes between current updated policy and old policy can be written as
\begin{align}
    r_t= \frac{\pi_t(a_t|s_t;\theta_{new})}{\pi_t(a_t|s_t;\theta_{old})}.
\end{align}
Thus, the new estimator can be modified as 
\begin{align}\label{eq:actor_estimator}
    L(\pi_{\theta_a}) = E_{\tau}[min(r_tA_t, clip(r_t, 1-\epsilon, 1+\epsilon)A_t)].
\end{align}
It can be found that if $A_t(s,a)$ is negative, the estimator is bounded by $(1-\epsilon)A_t$. On the other hand, when $A_t(s,a)$ is positive,  the estimator is at most $(1+\epsilon)A_t$.    

As an actor-critic algorithm, PPO learns the baseline by using a critic neural network. The loss function of the critic network is given by
\begin{align}\label{eq:critic_loss}
    Loss(s_t,\theta_v) = G_t - V(s_t;\theta_v),
\end{align}
where $V(s_t;\theta_v)$ is the value generated by the critic neural network and $\theta_v$ is the parameter of the critic network.
During each episode, the actor neural network optimize the estimator $L(\pi_{\theta_a})$ in \eqref{eq:actor_estimator} with respect to $\theta_a$ and minimize the loss function $Loss(s_t,\theta_v)$ with respect to $\theta_v$. 

\subsection{Overview of Our Approach}
A stateless deep RL approach has been taken to solve our optimization problems $\mathbf{P1}$ and $\mathbf{P2}$, in which RL agent (PPO network) at the AP generates a transmission probabilities matrix $\mathbb{P}$ based on given power level set $\mathcal{V}$; once $\mathbb{P}$ has been generated, AP broadcasts it to all devices in the network. AP knows how many devices exist in the network but do not know the type of each device. Devices then start to upload packets to the AP with corresponding transmission probabilities in $\mathbb{P}$. Without loss of generality, devices are aware of their own type, so the transmission probabilities can be decided by a device once $\mathbb{P}$ is received. After $T$ time slot, the AP calculates the reward, updates the actor critic network and generates new transmission probabilities matrix. The stateless deep RL problems can be formulated as Markov Decision Process (MDPs), consisting of two key elements:
\begin{itemize}
    \item actions: the transmission probabilities matrix $\mathbb{P}$.
    \item reward: the reward collected during $T$ time slots.(e.g., Minimum expected throughput, Geometric mean of the expected throughput).
\end{itemize}
\subsubsection{Actions}
Recall that for the transmission probabilities matrix, we need to guarantee $\sum_{m=1}^n\tau_{n,m}\leq 1, \forall n\leq M$. One way to generate reasonable $\mathbb{P}$ is to introduce the Beta distribution. Beta distribution defines on the interval $[0,1]$, in which two positive parameters $\alpha$ and $\beta$ control the shape of the distribution. For $n$-th type of devices, in order to generate $\tau_{n,m}, \forall 1\leq m\leq n$, we first generate continuous number $a_n$ and $b_n$ from the agent, then calculate $\alpha_n = \exp(a_n)$, $\beta_n = \exp(b_n)$ to ensure $\alpha_n$ and $\beta_n$ are larger than or equal to 0. With given $\alpha_n$ and $\beta_n$, the cumulative distribution function (CDF) of the beta distribution can be calculated easily. Let $F_n(x)$ be the value of the CDF at $x$, we have $F_n(0) = 0$ and $F_n(1) = 1$. The probability that $n$-th type of devices transmit by using $m$-th power level $\tau_{n,m}$ can be calculated as 
\begin{align}\label{eq:prob_cal}
    \tau_{n,m} = F_n\left(\frac{m}{n+1}\right)-F_n\left(\frac{m-1}{n+1}\right), \forall 1\leq m\leq n.
\end{align}
Thus, the probability that the $n$-th type of devices do not transmit is $1-F_n(\frac{n}{n+1})$.

\subsubsection{Reward}
In order to find the $\mathbb{P}$ that achieves the max-min fairness, the reward is defined as the following:
\begin{align}\label{eq:reward_max_min}
    R = \min \left(\mathbf{\overline{Th}}_{l}\right),  \qquad \forall l\in\mathcal{N}.
\end{align}

If the goal is to find the $\mathbb{P}$ which maximizes the geometric mean of the expected throughput for all devices, the reward is given by
\begin{align}\label{eq:reward_geo-mean}
    R = \left(\prod_{l\in \mathcal{N}} \mathbf{\overline{Th}}_{l}\right)^{1/|\mathcal{N}|}.
\end{align}

We summarize our proposed approach in $\mathbf{Algoithm 1}$. 

\begin{algorithm}\label{alg}
\SetAlgoLined
Goal: $\mathbb{P}$\;
 Initialize: $T$, $\mathcal{V}$, $\Gamma$\;
 \While{True}{
  1. Generates $\{a_1, b_1, a_2, b_2, \dots, a_M, b_M\}$ through actor neural network\;
  2. Calculate $\{\alpha_1, \beta_1, \alpha_2, \beta_2, \dots, \alpha_M, \beta_M\}$\;
  3. Generate $\mathbb{P}$ by using \eqref{eq:prob_cal}\;
  4. Broadcast $\mathbb{P}$ to all IoT devices in the network\;
  5. Wait $T$ time slots and calculate $\mathbf{\overline{Th}}_{l},\forall l\in\mathcal{N} $\;
  6. Calculate reward $R$ by using \eqref{eq:reward_max_min} or \eqref{eq:reward_geo-mean}\;
  7. Update weights of actor neural network and critic neural network by optimizing \eqref{eq:actor_estimator} and minimizing \eqref{eq:critic_loss} respectively;
  }
\caption{Agent Learning Process}
\end{algorithm}

\vspace{-9pt}

\section{Performance Evaluation} \label{sec:numerical}
\begin{figure*}[!htb]
  \centering
  \hfill
 \begin{minipage}{0.3\textwidth} 
    \centering
\subfloat[Total expected throughput as the reward function.]{\includegraphics[width=\textwidth]{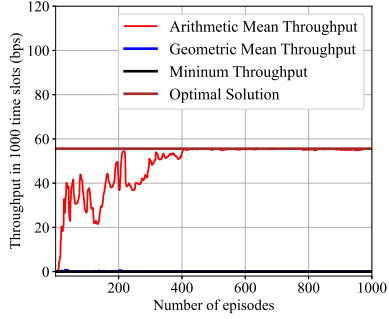}  
    \label{fig:to_th}}
  \end{minipage}
  \hfill
 \begin{minipage}{0.3\textwidth}
  	\centering
    \subfloat[Geometric mean of the expected throughput as the reward function.]{\includegraphics[width=\linewidth]{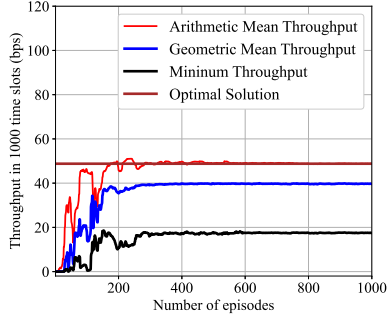}  
\label{fig:geo_th}}
  \end{minipage}
  \hfill
  \begin{minipage}{0.3\textwidth}
    \centering
    \subfloat[Minimum expected throughput as the reward function.]{\includegraphics[width=\linewidth]{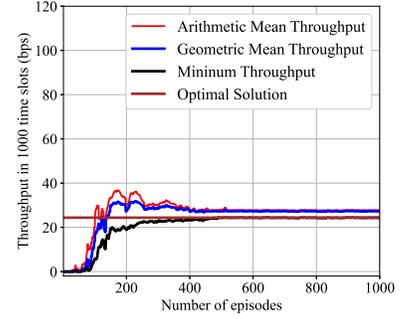} 
\label{fig:mm_th}}
  \end{minipage}
  \caption{Comparison between different reward functions with respect to the system performance.}
  \label{fig:th_compare}
  \hfill
\end{figure*} 

\begin{figure*}[!htb]
  \centering
  \hfill
 \begin{minipage}{0.3\textwidth} 
    \centering
\subfloat[Total expected throughput as the reward function.]{\includegraphics[width=\textwidth]{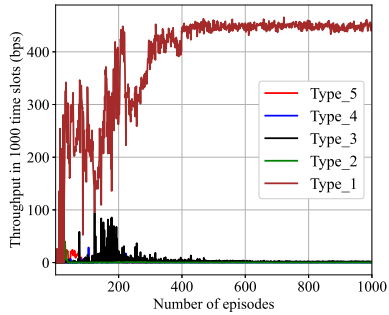}  
    \label{fig:total_type}}
  \end{minipage}
  \hfill
 \begin{minipage}{0.3\textwidth}
  	\centering
    \subfloat[Geometric mean of the expected throughput as the reward function.]{\includegraphics[width=\linewidth]{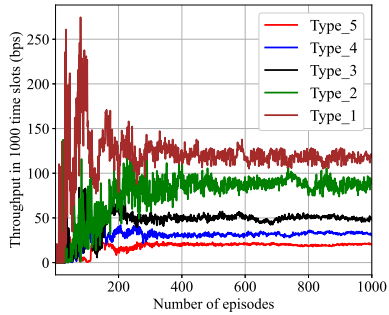}  
\label{fig:geo_type}}
  \end{minipage}
  \hfill
  \begin{minipage}{0.3\textwidth}
    \centering
    \subfloat[Minimum expected throughput as the reward function.]{\includegraphics[width=\linewidth]{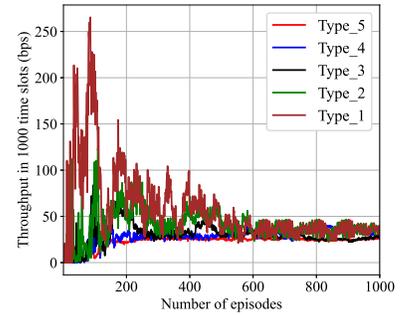} 
\label{fig:mm_type}}
  \end{minipage}
  \caption{Comparison between different reward functions with respect to the performance of different type of devices.}
  \label{fig:diff_type}
  \hfill
\end{figure*} 

The performance of a NOMA-based random access network with truncated channel inversion, 
which adopts the proposed deep RL approach has been evaluated in this section. 
The IoT network parameters we used in the simulation are $M = 5$, $T = 3000$, $V_{max} = 16$, $\{|\mathcal{N}_1|, |\mathcal{N}_2|, \dots, |\mathcal{N}_5|\} = \{5, 5, 8 , 10, 12\}$, $
\Gamma = 1$, $\delta^2 = 1$ if not otherwise specified. The power level set can be calculated as mentioned in Lemma. \ref{lemma:power_level}, in which $\mathcal{V} = [16, 8, 4, 2, 1]$. 
The parameters for the deep RL approach are shown in follow: learning rate of the actor and critic network are $1e^{-4}$, and the clipping parameter $\epsilon$ is $0.3$.


In Fig.~\ref{fig:th_compare}, the network performance such as the arithmetic mean of the expected throughput, geometric mean of the expected throughput, and the minimum expected throughput, are shown for the cases of (a) using total expected throughput as the reward function, (b) using geometric mean of the expected throughput as the reward function, and (c) using the minimum expected throughput as the reward function respectively. It can be intuitively seen that our approach has fast convergence speed and is capable to find the optimal solution. The corresponding average throughput for each type of devices are shown in Fig.~\ref{fig:diff_type}. 

For the case (a), after convergence, the geometric mean of the expected throughput and the minimum expected throughput are all zeros as shown in Fig.~\ref{fig:to_th}. It can also be observed in Fig.~\ref{fig:total_type} that only type 5 devices who can utilize all power levels are transmitting. This is due to the fact that the transmission of other type of devices will only increase the congestion probability of the channel, the best strategy to maximize the total expected throughput is to maximize the performance of the type 5 devices while other type of devices stop their transmission. 
When it comes to the case (b), the network performance and the average throughput for each type of devices are shown in Fig.~\ref{fig:geo_th} and Fig.~\ref{fig:geo_type} respectively. As can be observed that after convergence, all type of devices are capable to transmit while the total expected throughput is still relatively high. However, the average throughput of devices with more power levels to utilize is higher than the average throughput of devices with less power levels options to transmit. In other word, the network is not totally fair for all devices. In Fig.~\ref{fig:mm_th}, we plot the network performance of case (c). As shown in this figure, after convergence, the gap between arithmetic mean of the expected throughput, geometric mean of the expected throughput and the minimum expected throughput are small. This is because after maximize the minimum expected throughput, the average throughput of each type of devices are generally equal as shown in Fig.~\ref{fig:mm_type}. Otherwise, if the average throughput of a type of devices is smaller than the average throughput of other type of devices, the transmission probabilities of other type of devices should be decremented to promote the average throughput of the type of devices with minimum expected throughput, which helps to increase the minimum expected throughput of the network. It is worth noting that even thought the case (c) achieves fairness and improves the minimum throughput of the system, the total expected throughput is lower than that of case (b), which indicates that in order to achieve max-min fairness, the average throughput of devices with more power levels to select has been severely influenced.

To evaluate the effect of the power level set $\mathcal{V}$ on the system performance, we depict the reward of the proposed RL approach of case (b) and case (c) in Fig.~\ref{fig:com_geo} and Fig.~\ref{fig:com_mm} respectively. The results shown that better network performance could be achieved with more power levels available in the system.

\begin{figure}
\centering
    \includegraphics[width=0.3\textwidth]{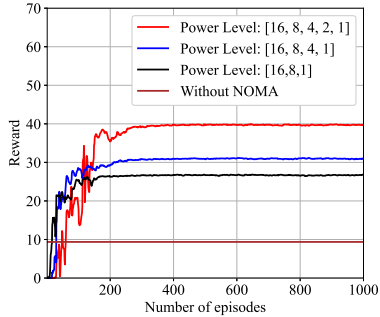} 
\caption{Reward of the RL approach with different $\mathcal{V}$ when case (b) is considered.} 
\label{fig:com_geo} 
\end{figure} 

\begin{figure}
    \centering
    \includegraphics[width=0.3\textwidth]{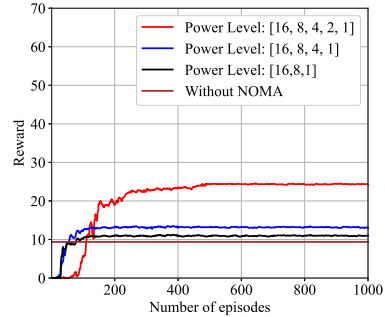}  
    \caption{Reward of the RL approach with different $\mathcal{V}$ when case (c) is considered.}
    \label{fig:com_mm} 
\end{figure} 


\vspace{-11pt}
\section{Conclusion} \label{sec:conclusion}
\vspace{-6pt}
In this paper, we introduced a novel deep RL approach for NOMA-based slotted ALOHA network with truncated channel inversion power control. The proposed approach enables the AP to tune the transmission probabilities matrix which guides the uplink transmission of IoT devices to improve the network performance. Instead of optimizing the total expected throughput which ignored the fairness in this heterogeneous scenario, two optimization problems have been formulated to maximize the geometric mean of the expected throughput and the minimum expected throughput for all devices. Extensive simulations shown that our approach helps us to find the optimal solutions of our optimization problems. In our future work, we will incorporate intelligent reflecting surface (IRS) in the NOMA-based random access network. 
\vspace{-7pt}
\section{Acknowledgement}
\vspace{-8pt}
This work is supported in part by NSF Career award ECCS1554576 and NSF grants CNS-1320736.
\vspace{-8pt}
\bibliographystyle{IEEEtran}
\bibliography{IEEEfull,ref}
\end{document}